\documentclass[twocolumn,showpacs,preprintnumbers,amsmath,amssymb,pra]{revtex4}

\usepackage{graphicx}
\usepackage{dcolumn}
\usepackage{bm}
\usepackage{mathrsfs}
\usepackage{amssymb}
\usepackage{amsmath}
\usepackage{amsthm}

\newtheorem{proposition}{Proposition}

\newtheorem{lemma}{Lemma}
\newtheorem{corollary}{Corollary}

\theoremstyle{definition}

\newtheorem{example}{Example}

\newcommand{\ave}[1]{\langle #1 \rangle}
\newcommand{\bra}[1]{\langle #1|}
\newcommand{\ket}[1]{| #1 \rangle }
\newcommand{\ip}[2]{{\langle #1|}{ #2 \rangle }}
\newcommand{\tr}[1]{{\rm tr}[#1]}
\newcommand{\be}{\begin{eqnarray}}
\newcommand{\ee}{\end{eqnarray}}
\def\Tr{{\rm tr}}

\newcommand{\cE}{{\cal E}}
\newcommand{\cG}{{\cal G}}

\newcommand{\cI}{{\cal I}}

\newcommand{\cF}{{\cal F}}

\newcommand{\cA}{{\cal A}}

\begin{document}

\title{Local two-qubit entanglement-annihilating channels}

\author{Sergey N. Filippov$^1$, Tom\'{a}\v{s} Ryb\'{a}r$^2$, and M\'{a}rio Ziman$^{2,3}$}
\affiliation{%
$^1$Moscow Institute of Physics and Technology, Moscow Region, Russia \\
$^2$Institute of Physics, Slovak Academy of
Sciences, Bratislava, Slovakia \\
$^3$Institute for Theoretical Physics, ETH Zurich, 8093 Zurich, Switzerland}

\date{\today}

\begin{abstract}
We address the problem of the robustness of entanglement of
bipartite systems (qubits) interacting with dynamically
independent environments. In particular, we focus
on characterization of
so-called local entanglement-annihilating two-qubit channels,
which set the maximum permissible noise level allowing to perform
entanglement-enabled experiments. The differences, but also subtle
relations between entanglement-breaking and local
entanglement-annihilating channels are emphasized. A detailed
characterization of latter ones is provided for a variety of
channels including depolarizing, unital, (generalized)
amplitude-damping, and extremal channels. We consider also the
convexity structure of local entanglement-annihilating qubit
channels and introduce a concept of entanglement-annihilation
duality.
\end{abstract}

\pacs{03.65.Ud, 03.67.Mn, 03.67.-a}

\maketitle

\section{\label{section:introduction} Introduction}

Flourishing field of quantum information theory is obliged to the
phenomenon of quantum entanglement~\cite{schrodinger} exhibited by
multipartite quantum systems. In the last decades many
entanglement-enabled applications of quantum states have been
developed and experimentally realized such as quantum key
distribution, dense coding, quantum teleportation, etc. (see the
detailed review~\cite{horodecki-review}). These quantum
information protocols operate efficiently providing that the
entanglement between the involved parties (Alice and Bob) is
preserved. However, during the protocols the involved systems
interact with an environment, which introduces (practically)
unavoidable noise. As a result of these influences, Alice and Bob
manipulate modified states whose entanglement can substantially
differ from the original one. It may even happen that the systems
become disentangled whatever state they start with. Under such
circumstances no entanglement-enabled application is
implementable.

For purposes of quantum communication protocols it is reasonable
to assume that the influence of environments of Alice and Bob are
independent. That is, the joint noise applied on the shared state
$\rho_{\rm in}$ is of the form $\cE_1\otimes\cE_2$, where
$\cE_1,\cE_2$ are local channels describing the interaction of
Alice's and Bob's subsystem, respectively, with their
environments. The question of our interest is the robustness of
the initial entanglement with respect to local noises, i.e. to
characterize the entanglement properties of states $\rho_{\rm
out}=(\cE_1 \otimes\cE_2) [\rho_{\rm in}]$. Different variations
of this problem of so-called entanglement dynamics were addressed
in a number of papers (see,
e.g.,~\cite{zyczkowski,sinayskiy,shan,scala,cui,altintas,ferraro,li2010,li2011,merkli}),
where the time evolution of two-qubit entanglement was studied for
different physical systems (initial state, types of interqubit
interactions and environments). Many researches have also paid
their attention to the phenomena known in the literature as sudden
death and sudden birth of the entanglement (see,
e.g.,~\cite{almeida,yu} and references therein). In contrast to
the studies, where the time evolution of the entanglement is
deduced from the time evolution of the state, an attempt to find a
direct relation (inequality) involving the initial and final
entanglements of an arbitrary bipartite two-qubit state in
presence of local noises has been undertaken in the
papers~\cite{deBrito,tiersch,konrad,zhang2010}.

Recently, a related concept of entanglement-annihilating channels
has been introduced~\cite{moravchikova}. These channels destroy
any quantum entanglement completely within the system they act on.
Following the paper~\cite{moravchikova}, we will refer to a local
two-qubit channel $\cE_1 \otimes \cE_2$ as
\textit{entanglement-annihilating} (EA) if the output state
$(\cE_1\otimes \cE_2) [\rho_{\rm in}]$ is separable for
\textit{all} input states $\rho_{\rm in}$. Therefore, the question
whether the channel is EA or not is a question whether the noise
level is acceptable for entanglement-enabled quantum applications
or not.

It is worth emphasizing the contrast between
entanglement-annihilating and entanglement-breaking channels. Let
us remind that a channel $\cE$ (acting on some system) is called
\textit{entanglement-breaking} (EB) if for all its extensions
$\cE\otimes\cI_{\rm anc}$ it annihilates the entanglement between
the system and the ancilla. In particular, a local two-qubit
channel $\cE_1 \otimes \cE_2$ is EB if the output state $\rho_{\rm
out}^{\rm +anc} = (\cE_1 \otimes \cE_2 \otimes \mathcal{I}_{\rm
anc}) [\rho_{\rm in}^{\rm +anc}]$ is disentangled with respect to
partitioning `1+2'$|$`anc' for any input state and any dimension
of the ancillary system. The entanglement-breaking channels and
their properties have been widely discussed in the literature
(see,
e.g.,~\cite{holevo,king,shor,ruskai-eb,horodecki-shor-ruskai,holevo08}).
As it was shown in Ref.~\cite{moravchikova} even if the channels
$\cE_1=\cE$ and $\cE_2=\cE$ are not EB, the channel
$\cE\otimes\cE$ can cancel any entanglement between quantum
subsystems in interest. It means that in order to fulfill an
entanglement-enabled protocol it does not suffice to know whether
the individual local influences are described by EB channel or
not, one has to resort to the concept of EA channel.

Our goal in this paper is to investigate in details the properties of
entanglement-annihilating channels for the simplest case of
two-qubits system. In Sec.~\ref{sec:properties}, we briefly review some
known properties of EA channels and add a new one for EA channels
of the form $\cE\otimes\cE$. Such channels naturally occur in
physical experiments, where two parties experience the same
influence from the environment and thus undergo the same
local transformation $\cE_1=\cE_2=\cE$. Although such channels do not
describe the general case, they are of significant physical
relevance. For the sake of convenience, if $\cE\otimes\cE$ is EA,
we will refer to a single-qubit channel $\cE$ as a
\emph{2-locally entanglement-annihilating channel} (2-LEA).

In order to demonstrate the difference between the EA and EB local
two-qubit channels, in Sec.~\ref{sec:depolarizing} we consider the
simplest and the most widely used model of quantum noise --
depolarizing channels. We find the overall noise level under which
the entanglement is destroyed for all input states or can be
preserved for some quantum states. Since depolarizing channels
belong to the class of unital channels, we further in
Sec.~\ref{sec:unital} focus our attention on this class. These
channels describe important physical processes that do not
increase purity of the states.

In Sec.~\ref{sec:non-unital} we move on to the
entanglement-annihilating behavior of non-unital channels. At
first, we address the question `which extremal single-qubit
channels $\cE_1$ and $\cE_2$ result in EA channels
$\cE_1\otimes\cE_2$?' Then, we consider amplitude-damping and
generalized amplitude-damping channels as the most prominent
representatives of non-unital channels. In Sec.~\ref{sec:extremal}
we remind that the set of all two-qubit EA channels (including
non-local ones) is convex~\cite{moravchikova}. This fact motivates
us to find EA-extremal channels and determine their position with
respect to the set of all two-qubit channels and its extreme
points. In Sec.~\ref{sec:duality} we find it quite interesting to
reveal and briefly outline an EA-duality between subsets of
local channels. Finally, we summarize the obtained
results in Sec.~\ref{sec:conclusions}.

\section{\label{sec:properties}Properties of EA-channels}

To begin with, we epitomize some basic properties of general
EA-channels found in Ref.~\cite{moravchikova}:

\begin{itemize}
\item[$1^{\circ}\,$] The set ${\sf T}_{\rm EA}$ of all EA
channels (including non-local ones) is convex.

\item[$2^{\circ}\,$] The channel is EA if and only if it
destroys entanglement of all pure input states.

\item[$3^{\circ}\,$] If $\cG_{12}$ is EA (local or non-local),
then $\cG_{12} \cdot \cF_{12}$ is EA for all two-qubit channels
$\cF_{12}$.

\item[$4^{\circ}\,$] The channels $\cE_1\otimes\cE_2$ and
$\cE_2\otimes\cE_1$ exhibit the same entanglement-annihilating
behavior.

\item[$5^{\circ}\,$] If $\cE_1$ or $\cE_2$ is EB, then
$\cE_1\otimes\cE_2$ is EA. This follows immediately from the
definitions of EA and EB channels.

\item[$6^{\circ}\,$] The channel $\cE\otimes\cI$ is EA if and only
if $\cE$ is EB.

\item[$7^{\circ}$] If $\cE$ is 2-LEA and $\cF$ is EB, then the
convex combination $\mu \cE + (1-\mu)\cF$ is 2-LEA, $\mu\in[0,1]$.

\end{itemize}

The last property was not shown in Ref.~\cite{moravchikova}. In
order to prove it recall the definition of 2-LEA channels from the
previous section. Suppose now the composite channel $\mu^2
\cE\otimes\cE + \mu(1-\mu) \cE\otimes\cF +\mu(1-\mu) \cF\otimes\cE
+ (1-\mu)^2 \cF\otimes\cF$. Then the channel $\cE\otimes\cE$ is EA
by definition of the 2-LEA channel $\cE$, the rest channels
$\cE\otimes\cF$, $\cF\otimes\cE$, and $\cF\otimes\cF$ are EA in
view of property~$5^{\circ}$. The convexity property~$1^{\circ}$
concludes the proof of property~$7^{\circ}$.

Is is worth mentioning that $\cE_1 \otimes \cE_2$ is EB if and
only if both $\cE_1$ and $\cE_2$ are EB. Thus, if the local
channel $\cE_1 \otimes \cE_2$ is EB then it is also EA by
property~$5^{\circ}$.

\section{\label{sec:depolarizing}Case study: Depolarizing channels}
The action of a depolarizing channel on $j$th qubit ($j=1,2$) is
defined as follows
\begin{equation}
\cE_j [X] = q_j X + (1-q_j) {\rm tr}[X] \frac{1}{2} I,
\end{equation}

\noindent where $q_j\in[-\frac{1}{3},1]$ and $I$ denotes the
identity operator. As a result of such noise the Bloch spheres of
individual qubits symmetrically shrink (in all directions). The
class of these channels was used in Ref.~\cite{moravchikova} to
show the existence of not entanglement-breaking 2-LEA channels. In
this section we will analyze the properties of depolarizing
channels of a slightly more general form $\cE_1\otimes\cE_2$. Each
channel $\cE_j$ is known to be EB if and only if $q_j \le
\frac{1}{3}$ (see, e.g.,~\cite{ruskai-eb}). If $q_j \le
\frac{1}{3}$ then the $j$-th qubit becomes disentangled from
arbitrary environment including the rest qubit. Therefore, the
two-qubit channel $\cE_1\otimes\cE_2$ is EB if and only if
simultaneously $q_1 \le \frac{1}{3}$ and $q_2 \le \frac{1}{3}$.

Let us find out when $\cE_1\otimes\cE_2$ is EA and destroys any
entanglement between qubits. We resort to property~$2^{\circ}$ and
consider pure input states $\omega = \ket{\psi} \bra{\psi}$. We
use the Schmidt decomposition of the state vector $\ket{\psi} =
\sqrt{p} \ket{\varphi\otimes\chi}+\sqrt{p_{\perp}}
\ket{\varphi_{\perp}\otimes\chi_{\perp}}$, where
$\{\ket{\varphi},\ket{\varphi_{\perp}}\}$ and
$\{\ket{\chi},\ket{\chi_{\perp}}\}$ are suitable orthornormal
bases of the first and the second qubit, respectively, $p$ and
$p_{\perp}$ are real nonnegative numbers such that
$p+p_{\perp}=1$.

Action of the two-qubit channel $\cE_1\otimes\cE_2$ on the state
$\omega$ yields
\begin{eqnarray}
&& \omega_{\rm out} = (\cE_1\otimes\cE_2)[\omega] = q_1 q_2 \omega
+ \frac{1}{2}
(1-q_1)q_2 I \otimes \omega_2 \nonumber\\
&& + \frac{1}{2} q_1 (1-q_2) \omega_1 \otimes I + \frac{1}{4}
(1-q_1) (1-q_2) I \otimes I, \qquad
\end{eqnarray}

\noindent with the reduced states $\omega_1 = p \ket{\varphi}
\bra{\varphi} + p_{\perp} \ket{\varphi_{\perp}}
\bra{\varphi_{\perp}}$ and $\omega_2 = p \ket{\chi} \bra{\chi} +
p_{\perp} \ket{\chi_{\perp}} \bra{\chi_{\perp}}$. According to the
Peres-Horodecki criterion~\cite{peres,horodecki96}, the output
state $\omega_{\rm out}$ is separable if the partially transposed
operator $\omega_{\rm out}^{\Gamma}$ is positive-semidefinite. The
condition $\omega_{\rm out}^{\Gamma} \ge 0$ reduces to
\begin{eqnarray}
\label{inequality}
\left|%
\begin{array}{cc}
  A+B & C \\
  C & A-B \\
\end{array}%
\right| = A^2-B^2-C^2 \ge 0,
\end{eqnarray}

\noindent where $A=1-q_1 q_2$, $B=(2p-1)(q_1 - q_2)$, and $C = 4
q_1 q_2 \sqrt{p p_{\perp}} = 4 q_1 q_2 \sqrt{p(1-p)}$. After
simplification we obtain
\begin{equation}
\nonumber (1+q_1q_2)(1-3q_1q_2) + 4 \left( p-\tfrac{1}{2}
\right)^2 \left[ 4 q_1^2 q_2^2 - (q_1-q_2)^2 \right] \ge 0.
\end{equation}
The channel in question is EA if this inequality holds for all $p
\in [0,1]$. If $2 |q_1 q_2| \geq |q_1 - q_2|$, then the minimum in
$p$ is achieved for $p=\frac{1}{2}$ and we end up with the
inequality $(1+q_1q_2)(1-3q_1q_2) \ge 0$. Taking into account that
$-\tfrac{1}{3}\leq q_1,q_2\leq 1$, this inequality holds whenever
$q_1q_2\leq \frac{1}{3}$. If $2 |q_1 q_2| < |q_1 - q_2|$, then the
minimum is achieved for $p=0$, or $p=1$, and for this case the
inequality takes the form $(1-q_1^2)(1-q_2^2)\geq 0$, which is
always satisfied for all allowed values of $q_1,q_2$.

In summary, the channel $\cE_1\otimes\cE_2$ is entanglement
annihilating if and only if $q_1 q_2 \leq \tfrac{1}{3}$. On the
contrary, such a two-qubit channel is EB only if and only if
simultaneously $q_1 \le \frac{1}{3}$ and $q_2 \le \frac{1}{3}$
(see Fig.~\ref{fig:depolarizing}).

\begin{figure}
\includegraphics[width=8.0cm]{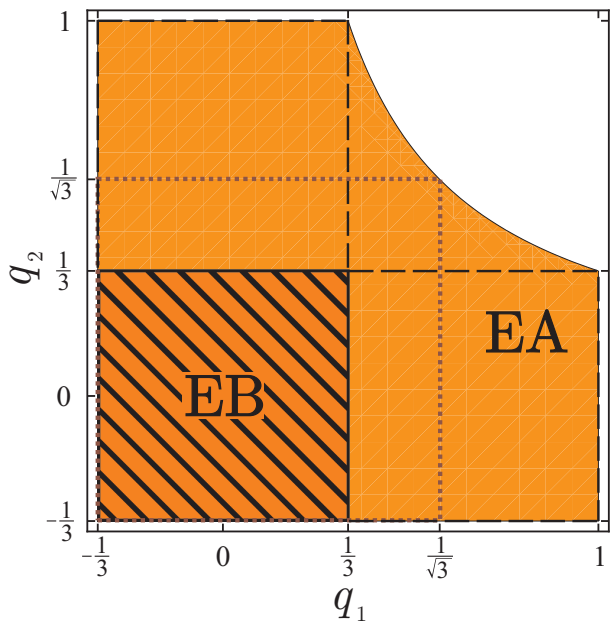}
\caption{\label{fig:depolarizing} Entanglement-annihilating (EA)
and entanglement-breaking (EB) local two-qubit channels $\cE_1
\otimes \cE_2$ composed of depolarizing channels $\cE_j [X] = q_j
X + (1-q_j) {\rm tr}[X]\frac{1}{2} I$, $j=1,2$. Dashed lines
depict regions where at least one of the channels $\cE_1$ or
$\cE_2$ is EB, and $\cE_1 \otimes\cE_2$ is EA (see
property~$5^\circ$). Dotted line depicts the region, where both
$\cE_1^2$ and $\cE_2^2$ are EB, i.e. $\cE_1 \otimes \cE_2$ is EA
(see Corollary 1).}
\end{figure}

\section{\label{sec:unital} Unital channels}
In this section we will focus our investigation on the class of
qubit unital channels. By definition a channel is unital if it
preserves the identity operator, i.e. $\cE[I]=I$. As it was shown
in Ref.~\cite{ruskai} any such channel can be expressed as a
diagonal matrix (acting on Bloch vectors) in a properly chosen
basis of self-adjoint operators $\{\sigma_0 \equiv
I,\sigma_1,\sigma_2,\sigma_3\}$, where
$\tr{\sigma_j\sigma_k}=2\delta_{jk}$ and $\delta_{jk}$ is the
conventional Kronecker delta symbol. Channels of this (diagonal)
form are also known as Pauli channels.

It follows that the channel $\cE_1 \otimes \cE_2$, where $\cE_1$
and $\cE_2$ are unital single-qubit channels, has a diagonal
matrix representation in the basis of individual Pauli operators
$\{\sigma_m \otimes \sigma'_n\}_{m,n=0}^3$. In particular, the
entries of the matrix representation of $\cE_1\otimes\cE_2$ read
\[
\frac{1}{4}\tr {\sigma_k\otimes\sigma'_l \, (\cE_1\otimes\cE_2)
[\sigma_m\otimes\sigma'_n]} = \lambda_m \lambda'_n
\delta_{km}\delta_{ln},
\]

\noindent where $\lambda_0=\lambda'_0=1$ in view of the
trace-preserving property, $\{\lambda_m\}_{m=1}^3$ and
$\{\lambda'_n\}_{n=1}^3$ are singular values of the channels
$\cE_1$ and $\cE_2$, respectively. The output state of the local
two-qubit unital channel $\cE_1\otimes\cE_2$ takes the form
\begin{equation}
\label{rho-out} \rho_{\rm out} = \frac{1}{4}\sum_{m,n=0}^{3}
\lambda_m \lambda'_n {\rm tr} [ \rho_{\rm in} \sigma_{m} \otimes
\sigma'_{n}] \sigma_{m} \otimes \sigma'_{n}.
\end{equation}

For each unital qubit channel $\cE$ we introduce a unital map
$\overline{\cE}$ with $\lambda_0 = 1$ and singular values
$\{-\lambda_m\}_{m=1}^3$, where $\{\lambda_m\}_{m=1}^3$ are
singular values of the original channel $\cE$. Note that the map
$\overline{\cE}$ is positive and trace preserving but not
necessarily completely positive.  This means that
$\overline{\cE}[\rho] \ge 0$ for all qubit density operators
$\rho$, whereas $(\overline{\cE}\otimes \cI) [\rho_{\rm in}]$ can,
in principle, have negative eigenvalues for some two-qubit density
operators $\rho_{\rm in}$.

The entanglement-annihilating behavior of local two-qubit unital
channels $\cE_1\otimes\cE_2$ is governed by the following Lemma.

\begin{lemma}
\label{lemma1} Let $\cE_1$ and $\cE_2$ be unital qubit channels.
The two-qubit channel $\cE_1\otimes\cE_2$ is EA if and only if the
maps $\overline{\cE_1} \otimes \cE_2$ and $\cE_1 \otimes
\overline{\cE_2}$ are positive.
\end{lemma}

\begin{proof} Separability of the output state (\ref{rho-out}) of the channel $\cE_1 \otimes
\cE_2$ can be checked by the reduction
criterion~\cite{horodecki99}, which turns out to be a necessary
and sufficient separability condition for two-qubit systems.
Assuming that the channel $\cE_1\otimes\cE_2$ is
entanglement-annihilating, the two-qubit state $\rho_{\rm out}$ is
separable, thus, in accordance with the reduction crierion the
following conditions hold
\begin{equation}
\label{reduction-criterion} \Tr_2 [\rho_{\rm out}] \otimes I -
\rho_{\rm out} \ge 0 ~~ {\rm and}~~ I \otimes \Tr_1 [\rho_{\rm
out}] - \rho_{\rm out} \ge 0,
\end{equation}

\noindent where $\Tr_1 [ \cdot ]$ and $\Tr_2 [ \cdot ]$ denote
partial traces over the first and the second qubit, respectively.
Since,
\begin{eqnarray}
&& (\cE_1\otimes\overline{\cE_2}) [\varrho_{\rm in}]=\Tr_2
[\rho_{\rm
out}] \otimes I - \rho_{\rm out}\,; \nonumber\\
&& (\overline{\cE_{1}}\otimes{\cE_2}) [\varrho_{\rm in}] =I
\otimes \Tr_1 [\rho_{\rm out}] - \rho_{\rm out}\,, \nonumber
\end{eqnarray}

\noindent the above separability conditions are equivalent with
positivity of the maps $\overline{\cE_1} \otimes \cE_2$ and $\cE_1
\otimes \overline{\cE_2}$, respectively.
\end{proof}

Before we explore the consequences of Lemma~\ref{lemma1} and
derive some properties of local two-qubit unital channels
$\cE_1\otimes\cE_2$, let us make some remarks about single-qubit
unital channels. A qubit unital map $\cE$ with singular values
$\{\lambda_m\}_{m=1}^3$ is indeed a channel (i.e. completely
positive trace-preserving map) if $1+\lambda_1+\lambda_2+\lambda_3
\ge 0$, $1+\lambda_1-\lambda_2-\lambda_3 \ge 0$,
$1-\lambda_1+\lambda_2-\lambda_3 \ge 0$, and
$1-\lambda_1-\lambda_2+\lambda_3 \ge 0$. These four inequalities
define a tetrahedron in the conventional reference frame
$(\lambda_1,\lambda_2,\lambda_3)$ in $\mathbb{R}^3$ (see
Fig.~\ref{fig:sphere}). The channel $\cE$ is known to be EB if and
only if $|\lambda_1|+|\lambda_2|+|\lambda_3| \le 1$~\cite{ruskai}.
This inequality corresponds to the octahedron in
Fig.~\ref{fig:sphere}.

Let us note that $\overline{\cE}$ is a quantum channel if and only
if $\cE$ is entanglement-breaking. That is, Lemma~\ref{lemma1}
guarantees that $\cE_1\otimes\cE_2$ is EA if $\cE_1$ or $\cE_2$ is
EB, which is in agreement with the property~$5^{\circ}$.
Nevertheless, the channel $\cE_1\otimes\cE_2$ can be EA even if
neither $\cE_1$ nor $\cE_2$ is EB. The following proposition gives
a sufficient condition for $\cE_1\otimes\cE_2$ being
entanglement-annihilating.

\begin{proposition}\label{proposition1}
Suppose $\cE_1,\cE_2$ are unital qubit channels such that
$\cE_1^2$ and $\cE_2^2$ are entanglement-breaking channels, i.e.
$\sum_{m=1}^3 \lambda_m^2 \le 1$ and $\sum_{n=1}^3 {\lambda'_n}^2
\le 1$. Then $\cE_1\otimes\cE_2$ is an entanglement-annihilating
channel.
\end{proposition}

\begin{proof}
Consider the map $\overline{\cE_1}\otimes\cE_2$. Let us
demonstrate that $(\overline{\cE_1}\otimes\cE_2)[\rho_{\rm in}]
\ge 0$ for all two-qubit input states $\rho_{\rm in}$. In view of
convexity of the state space it suffices to show that
$(\overline{\cE_1}\otimes\cE_2)[\ket{\psi}\bra{\psi}] \ge 0$ for
all pure two-qubit states $\ket{\psi}$.

Any state $\ket{\psi}$ is given by its Schmidt decomposition
$\sqrt{p} \ket{\varphi\otimes\chi}+\sqrt{p_{\perp}}
\ket{\varphi_{\perp}\otimes\chi_{\perp}}$, where $p,p_{\perp} \ge
0$ and $p+p_{\perp}=1$, the orthonormal basis
$\{\ket{\varphi},\ket{\varphi_{\perp}}\}$ can be parameterized by
the angles $\theta\in[0,\pi]$ and $\phi\in[0,2\pi]$ as follows:
\begin{eqnarray}
&& \label{basis-vector} \ket{\varphi} = \left(%
\begin{array}{c}
  \cos (\theta \!/ 2) \exp (-i \phi \!/ 2) \\
  \sin (\theta \!/ 2) \exp (i \phi \!/ 2) \\
\end{array}%
\right), \\
&& \label{basis-vector-perp} \ket{\varphi_{\perp}} = \left(%
\begin{array}{l}
  -\sin (\theta \!/ 2) \exp (-i \phi \!/ 2) \\
  \cos (\theta \!/ 2) \exp (i \phi \!/ 2) \\
\end{array}%
\right),
\end{eqnarray}

\noindent and the basis $\{\ket{\chi},\ket{\chi_{\perp}}\}$ is
obtained from formulas
(\ref{basis-vector})--(\ref{basis-vector-perp}) by replacing
$\ket{\varphi} \rightarrow \ket{\chi}$, $\ket{\varphi_\perp} \rightarrow
\ket{\chi_{\perp}}$, $\theta \rightarrow \theta'$, and $\phi
\rightarrow \phi'$.

The map $\overline{\cE_1}\otimes\cE_2$ transforms
$\ket{\psi}\bra{\psi}$ into the operator
\begin{eqnarray}
&& \label{output-unital} \!\!\!\!\! \tfrac{1}{4} \big\{ I
\!\otimes\! I \! - \! ({\bf n} \!\cdot\! \boldsymbol{\sigma}\!)
\!\otimes\! ({\bf n}'\! \!\cdot\! \boldsymbol{\sigma}'\!) \! - \!
(p \!-\! p_{\perp}) \left[ ({\bf n} \!\cdot\!
\boldsymbol{\sigma}\!) \!\otimes\! I - I\!\otimes\! ({\bf n}'\!
\!\cdot\!
\boldsymbol{\sigma}'\!) \right] \nonumber\\
&& \!\!\!\!\! -  2\sqrt{pp_{\perp}} \left[ ({\bf k} \!\cdot\!
\boldsymbol{\sigma}\!) \!\otimes\! ({\bf k}'\! \!\cdot\!
\boldsymbol{\sigma}'\!) \! - \! ({\bf l} \!\cdot\!
\boldsymbol{\sigma}\!) \!\otimes\! ({\bf l}'\! \!\cdot\!
\boldsymbol{\sigma}'\!) \right] \big\},
\end{eqnarray}

\noindent where $({\bf n} \cdot \boldsymbol{\sigma}) = n_1
\sigma_1 + n_2 \sigma_2 + n_3 \sigma_3$ and vectors ${\bf n},{\bf
k},{\bf l} \in \mathbb{R}^3$ are expressed through singular values
$\{\lambda_m\}_{m=1}^3$ of the channel $\cE_1$ by formulas
\begin{eqnarray}
&& \label{vector-3d-n} {\bf n} = (\lambda_1 \cos\phi
\sin\theta,\lambda_2 \sin\phi
\sin\theta, \lambda_3 \cos\theta), \\
&& \label{vector-3d-k} {\bf k} = (-\lambda_1 \cos\phi
\cos\theta,-\lambda_2 \sin\phi
\cos\theta, \lambda_3 \sin\theta), \\
&& \label{vector-3d-l} {\bf l} = (\lambda_1 \sin\phi,-\lambda_2
\cos\phi,0).
\end{eqnarray}

\noindent The vectors ${\bf n}',{\bf k}',{\bf l}'$ are obtained
from (\ref{vector-3d-n}), (\ref{vector-3d-k}),
(\ref{vector-3d-l}), respectively, by replacing
$\boldsymbol{\lambda} \rightarrow \boldsymbol{\lambda}'$, $\theta
\rightarrow \theta'$, and $\phi \rightarrow \phi'$.

Both sets of vectors $\{{\bf n}, {\bf k}, {\bf l}\}$ and $\{{\bf
n}', {\bf k}', {\bf l}'\}$ have a particular property
\begin{eqnarray}
&& \label{n-k-l-sum} |{\bf n}|^2 + |{\bf k}|^2 + |{\bf l}|^2 =
\textstyle{
\sum_{m=1}^3 \lambda_m^2 } \le 1, \\
&& \label{n-k-l-prime-sum} |{\bf n}'|^2 + |{\bf k}'|^2 + |{\bf
l}'|^2 = \textstyle{ \sum_{n=1}^3 {\lambda'_n}^2 } \le 1
\end{eqnarray}

\noindent thanks to the statement of the proposition.

The output state (\ref{output-unital}) is positive semi-definite
if and only if its average $\ave{\rho_{\rm out}} \ge 0$ for all
two-qubit states. In other words, we want to show that the
inequality
\begin{eqnarray}
&& \label{output-ave-ineq-1} \!\!\!\!\! \big\langle I \!\otimes\!
I \! - \! ({\bf n} \!\cdot\! \boldsymbol{\sigma}\!) \!\otimes\!
({\bf n}'\! \!\cdot\! \boldsymbol{\sigma}'\!) \big\rangle  \ge  (p
\!-\! p_{\perp}) \big\langle  ({\bf n} \!\cdot\!
\boldsymbol{\sigma}\!) \!\otimes\! I - I\!\otimes\! ({\bf n}'\!
\!\cdot\!
\boldsymbol{\sigma}'\!) \big\rangle \nonumber\\
&& \!\!\!\!\! + 2\sqrt{pp_{\perp}} \big\langle ({\bf k} \!\cdot\!
\boldsymbol{\sigma}\!) \!\otimes\! ({\bf k}'\! \!\cdot\!
\boldsymbol{\sigma}'\!) \! - \! ({\bf l} \!\cdot\!
\boldsymbol{\sigma}\!) \!\otimes\! ({\bf l}'\! \!\cdot\!
\boldsymbol{\sigma}'\!) \big\rangle
\end{eqnarray}

\noindent holds true for any averaging state. Since
$(p-p_{\perp})^2 + (2\sqrt{p p_{\perp}})^2 = (p+p_{\perp})^2 = 1$,
one can treat $(p-p_{\perp})$ as $\cos \alpha$ and $2\sqrt{p
p_{\perp}}$ as $\sin \alpha$. Due to the fact that $\max_{\alpha}
(A \cos\alpha + B \sin\alpha) = \sqrt { A^2 + B^2}$, the
inequality (\ref{output-ave-ineq-1}) if fulfilled whenever
\begin{eqnarray}
&& \label{output-ave-ineq-2} \big\langle I \!\otimes\! I \! - \!
({\bf n} \!\cdot\! \boldsymbol{\sigma}\!) \!\otimes\! ({\bf n}'\!
\!\cdot\! \boldsymbol{\sigma}'\!) \big\rangle^2  \ge  \big\langle
({\bf n} \!\cdot\! \boldsymbol{\sigma}\!) \!\otimes\! I -
I\!\otimes\! ({\bf n}'\! \!\cdot\!
\boldsymbol{\sigma}'\!) \big\rangle^2 \nonumber\\
&& + \big\langle ({\bf k} \!\cdot\! \boldsymbol{\sigma}\!)
\!\otimes\! ({\bf k}'\! \!\cdot\! \boldsymbol{\sigma}'\!) \! - \!
({\bf l} \!\cdot\! \boldsymbol{\sigma}\!) \!\otimes\! ({\bf l}'\!
\!\cdot\! \boldsymbol{\sigma}'\!) \big\rangle^2
\end{eqnarray}

\noindent or, equivalently,
\begin{eqnarray}
&& \label{output-ave-ineq-3} \!\!\!\!\!\! \big\langle \big( I
\!+\! ({\bf n} \!\cdot\! \boldsymbol{\sigma}\!) \big) \!\otimes\!
\big( I \!-\! ({\bf n}'\! \!\cdot\! \boldsymbol{\sigma}'\!) \big)
\big\rangle \big\langle \big( I \!-\! ({\bf n} \!\cdot\!
\boldsymbol{\sigma}\!) \big) \!\otimes\! \big( I \!+\! ({\bf n}'\!
\!\cdot\!
\boldsymbol{\sigma}'\!) \big) \big\rangle \ge \nonumber\\
&& \!\!\!\!\!\! \big\langle ({\bf k} \!\cdot\!
\boldsymbol{\sigma}\!) \!\otimes\! ({\bf k}'\! \!\cdot\!
\boldsymbol{\sigma}'\!) \! - \! ({\bf l} \!\cdot\!
\boldsymbol{\sigma}\!) \!\otimes\! ({\bf l}'\! \!\cdot\!
\boldsymbol{\sigma}'\!) \big\rangle^2.
\end{eqnarray}

Taking into account that $\ket{\varphi \otimes \chi}$,
$\ket{\varphi \otimes \chi_{\perp}}$, $\ket{\varphi_{\perp}
\otimes \chi}$, and $\ket{\varphi_{\perp} \otimes \chi_{\perp}}$
are all eigenvectors of the operators in the left hand side of
(\ref{output-ave-ineq-3}) and using the Cauchy-Schwarz inequality
(CS) in the form $\langle X^{\dag} X \rangle\langle Y^{\dag} Y
\rangle\ge |\langle X^{\dag} Y\rangle|^2$, one can readily see
that
\begin{eqnarray}
&& \label{inequalities}\!\!\!\!\!\! \big\langle \big( I \!+\!
({\bf n} \!\cdot\! \boldsymbol{\sigma}\!) \big) \!\otimes\! \big(
I \!-\! ({\bf n}'\! \!\cdot\! \boldsymbol{\sigma}'\!) \big)
\big\rangle \big\langle \big( I \!-\! ({\bf n} \!\cdot\!
\boldsymbol{\sigma}\!) \big) \!\otimes\! \big( I \!+\! ({\bf n}'\!
\!\cdot\!
\boldsymbol{\sigma}'\!) \big) \big\rangle \!\!\! \underset{\rm CS}{\ge} \nonumber\\
&& \!\!\!\!\!\! (1-|{\bf n}|^2)(1-|{\bf n}'|^2) \underset{\rm
(\ref{n-k-l-sum}),(\ref{n-k-l-prime-sum})}{\ge} (|{\bf k}|^2 +
|{\bf l}|^2)(|{\bf k}'|^2 + |{\bf l}'|^2) {\ge} \nonumber\\
&& \!\!\!\!\!\! \big( |{\bf k}| |{\bf k}'|+|{\bf l}| |{\bf l}'|
\big)^2 \ge \big\langle ({\bf k} \!\cdot\! \boldsymbol{\sigma}\!)
\!\otimes\! ({\bf k}'\! \!\cdot\! \boldsymbol{\sigma}'\!) \! - \!
({\bf l} \!\cdot\! \boldsymbol{\sigma}\!) \!\otimes\! ({\bf l}'\!
\!\cdot\! \boldsymbol{\sigma}'\!) \big\rangle^2.
\end{eqnarray}

Thus,
(\ref{n-k-l-sum})$\wedge$(\ref{n-k-l-prime-sum})$\Rightarrow$(\ref{inequalities})$\Rightarrow$(\ref{output-ave-ineq-3})$\Leftrightarrow$(\ref{output-ave-ineq-2})$\Rightarrow$(\ref{output-ave-ineq-1})$\Rightarrow$
$\overline{\cE_1}\otimes\cE_2$ is a positive map. In the same way,
$\cE_1\otimes\overline{\cE_2}$ is also a positive map. According
to Lemma~\ref{lemma1}, the channel $\cE_1 \otimes \cE_2$ is EA.
\end{proof}

Note that Proposition~\ref{proposition1} provides the sufficient
but not necessary condition for the channel $\cE_1 \otimes \cE_2$
to be EA. For instance, in case of two depolarizing channels
$\cE_1$ and $\cE_2$ from Sec.~\ref{sec:depolarizing}, the channels
$\cE_1^2$ and $\cE_2^2$ are EB if $q_1 \le \frac{1}{\sqrt{3}}$ and
$q_2 \le \frac{1}{\sqrt{3}}$, respectively. It means that
$\cE_1\otimes\cE_2$ is EA if $q_1,q_2 \le \frac{1}{\sqrt{3}}$. The
corresponding area of parameters $(q_1,q_2)$ is depicted in
Fig.~\ref{fig:depolarizing}, which illustrates the `power' of
Proposition~\ref{proposition1}. However, in case of identical
environments, i.e. $\cE_1=\cE_2=\cE$, the use of
Proposition~\ref{proposition1} provides sufficient and necessary
condition for $\cE\otimes\cE$ to be EA.

\begin{proposition}
\label{proposition2}
A unital qubit channel $\cE$ is 2-LEA if and only if the channel
$\cE^2$ is EB, i.e. $\lambda_1^2+\lambda_2^2+\lambda_3^2 \le 1$.
\end{proposition}

\begin{proof}
The condition \{$\cE^2$ is EB\} is sufficient due to
Proposition~\ref{proposition1}. It turns out to be also necessary
if we consider Bell states, e.g., $\ket{\psi_{+}} =
\frac{1}{\sqrt{2}} (\ket{00}+\ket{11})$. It is readily seen that
the operator $\left( (\cE\otimes\cE) [\ket{\psi}\bra{\psi}]
\right)^{\Gamma}$ is positive semi-definite if and only if
$\lambda_1^2+\lambda_2^2+\lambda_3^2 \le 1$. On the other hand,
$\{\lambda_i^2\}_{i=1}^{3}$ are eigenvalues of $\cE^2$ and their
sum is less or equal than 1, i.e. $\cE^2$ is EB~\cite{ruskai-eb}.
\end{proof}

An illustration of 2-LEA unital channels $\cE$ in the conventional
reference frame $\{\lambda_1,\lambda_2,\lambda_3\}$ is presented
in Fig.~\ref{fig:sphere}. Points outside the tetrahedron do not
satisfy the complete positivity of $\cE$. Points outside the
sphere do not correspond to 2-LEA channels, which can be easily
checked by Bell states. Both the sphere and tetrahedron comprise
an octahedron of EB single-qubit channels $\cE$ with $ |\lambda_1|
+ |\lambda_2|+|\lambda_3| \le 1$. Since intersection of the sphere
and tetrahedron is a convex body, we have just revealed an
interesting property:

\begin{proposition}
\label{proposition3}
The set of 2-LEA unital qubit channels
is convex, i.e. if $\cE_1$ and $\cE_2$ are 2-LEA, then
$\mu\cE_1+(1-\mu)\cE_2$ is also 2-LEA for all $\mu\in[0,1]$.
\end{proposition}

\begin{proof}
Despite this fact is evident from geometrical consideration, it
also follows from the relation $\big(\mu\cE_1+(1-\mu)\cE_2\big)
\otimes \big(\mu\cE_1+(1-\mu)\cE_2 \big) = \mu^2 \cE_1\otimes\cE_1
+ \mu(1-\mu)\cE_1\otimes\cE_2 + (1-\mu)\mu \cE_2\otimes\cE_1 +
(1-\mu)^2 \cE_2\otimes\cE_2$. As channels $\cE_1\otimes\cE_1$ and
$\cE_2\otimes\cE_2$ are EA by a statement of the involved
Proposition, the channels $\cE_1^2$ and $\cE_2^2$ are EB by
Proposition~\ref{proposition2}. Due to
Proposition~\ref{proposition1} both $\cE_1\otimes\cE_2$ and
$\cE_2\otimes\cE_1$ are EA. Then we use property~$1^{\circ}$ to
conclude the proof.
\end{proof}

One more property immediately follows from Propositions
\ref{proposition1} and \ref{proposition2}:

\begin{corollary}
If $\cE_1$ and $\cE_2$ are unital
2-LEA qubit channels, then $\cE_1\otimes\cE_2$ is EA.
\end{corollary}

We can make an interesting observation from Fig.~\ref{fig:sphere},
namely, phase-damping channels (forming the edges of the
tetrahedron) preserve entanglement unless they contract the whole
Bloch sphere into a line.

\begin{figure}
\includegraphics[width=8.50cm]{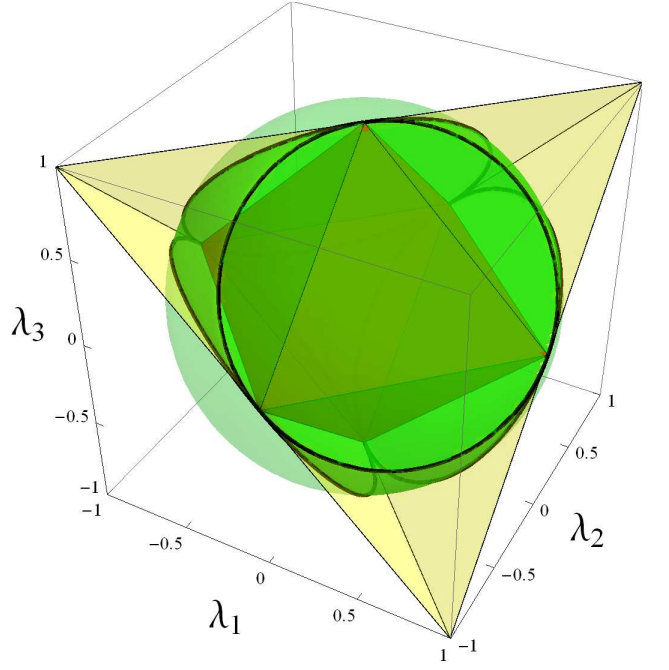}
\caption{\label{fig:sphere} Intersection of the tetrahedron and
the ball is a set of 2-LEA unital qubit channels $\cE$, i.e. such
channels $\cE$ that $\cE\otimes\cE$ is EA. The octahedron inside
the sphere represents EB qubit channels.}
\end{figure}

Let us now analyze unital channels $\cE_1 \otimes \cE_2$ such that
neither $\cE_1$, nor $\cE_2$ is a 2-LEA channel (i.e. both $\cE_1$
and $\cE_2$ are outside the sphere in Fig.~\ref{fig:sphere}).
Surprisingly, it turns out that $\cE_1 \otimes \cE_2$ can still be
EA as it is demonstrated by the following example.

\begin{example}
Suppose $\cE_1={\rm diag}\{1,\frac{1}{20},\frac{1}{20},1\}$ and
$\cE_2={\rm diag}\{1,\frac{2}{3},\frac{2}{3},\frac{2}{3}\}$. Let
us demonstrate that $\cE_1 \otimes \cE_2$ is EA although neither
$\cE_1$, nor $\cE_2$ is 2-LEA. We note that the channel in
question can be represented in the form $\cE_1 \otimes \cE_2 = \cG
\cdot \cF$, where $\cG = \frac{3}{4} \cG_1\otimes\cI + \frac{1}{4}
\cI\otimes\cG_2$, $\cG_1={\rm diag}\{1,0,0,1\}$, $\cG_2={\rm
diag}\{1,\frac{1}{3},\frac{1}{3},\frac{1}{3}\}$, and the two-qubit
map $\cF$ is defined by its matrix representation $\frac{1}{4}\tr
{\sigma_k\otimes\sigma'_l \cF[\sigma_m\otimes\sigma'_n]} = F_{mn}
\delta_{km}\delta_{ln}$ with $F_{0n}=F_{3n}=(4+\delta_{n0})/5$,
$F_{1n}=F_{2n}=(2-\delta_{n0})/5$. By considering the Choi matrix
of the map $\cF$, it is not hard to see that $\cF$ is a channel
indeed, i.e. completely positive trace-preserving map. As both
$\cG_1$ and $\cG_2$ are EB, the channels $\cG_1\otimes\cI$ and
$\cI\otimes\cG_2$ are EA by property $5^{\circ}$. Hence $\cG$ is
EA by property $1^{\circ}$ and $\cG \cdot \cF$ is EA by property
$3^{\circ}$. The equality $\cE_1 \otimes \cE_2 = \cG \cdot \cF$
completes the proof.
\end{example}

A natural question to ask is under which conditions on $\cE_1$ and
$\cE_2$ the channel $\cE_1\otimes\cE_2$ is entanglement-annihilating.
The following proposition  formulates a sufficient condition for the
converse statement.

\begin{proposition}
\label{proposition4}
Consider qubit unital channels $\cE_1$ and $\cE_2$ with singular
values $\boldsymbol{\lambda} = (\lambda_1,\lambda_2,\lambda_3)$
and $\boldsymbol{{\lambda}'} =
(\lambda'_1,\lambda'_2,\lambda'_3)$, respectively. If
$\boldsymbol{\lambda} \cdot \boldsymbol{{\lambda}'}
>1$, then $\cE_1\otimes\cE_2$ is not entanglement-annihilating channel.
\end{proposition}

\begin{proof} The channel
$\cE_1\otimes\cE_2$ preserves entanglement of the Bell state
$\ket{\psi_{+}} = \frac{1}{\sqrt{2}}(\ket{00} + \ket{11})$ when
the matrix
\begin{eqnarray}
&& \!\!\!\!\! \rho_{\rm out}^{\Gamma} = \left( (\cE_1\otimes\cE_2) [\ket{\psi_{+}} \bra{\psi_{+}}] \right)^{\Gamma} = \nonumber\\
&& \!\!\!\!\! \frac{1}{4} \left(%
\begin{array}{cccc}
  1+\lambda_3 \lambda'_3 & 0 & 0 & \lambda_1 \lambda'_1 \!\!-\!\! \lambda_2 \lambda'_2\\
  0 & 1-\lambda_3 \lambda'_3 & \lambda_1 \lambda'_1 \!\!+\!\! \lambda_2 \lambda'_2 & 0 \\
  0 & \lambda_1 \lambda'_1 \!\!+\!\! \lambda_2 \lambda'_2 & 1-\lambda_3 \lambda'_3 & 0 \\
  \lambda_1 \lambda'_1 \!\!-\!\! \lambda_2 \lambda'_2 & 0 & 0 & 1+\lambda_3 \lambda'_3 \\
\end{array}%
\right) \nonumber
\end{eqnarray}

\noindent has negative eigenvalues. It takes place if $\lambda_1
\lambda'_1 + \lambda_2 \lambda'_2 + \lambda_3 \lambda'_3 \equiv
\boldsymbol{\lambda} \cdot \boldsymbol{\lambda'} >1$. In view of
the Peres-Horodecki criterion~\cite{peres,horodecki96}, the output
state $\rho_{\rm out}$ remains entangled.
\end{proof}

Let us note that this proposition is very efficient, e.g., for depolarizing
channels from Sec.~\ref{sec:depolarizing} we immediately have
not-EA channel $\cE_1\otimes\cE_2$ if $q_1 q_2 > \frac{1}{3}$.
Remind that $q_1 q_2 = \frac{1}{3}$ is a boundary between EA and
not-EA behavior of depolarizing channels.

Our goal was to give a complete picture of unital two-qubit
channels of the form $\cE_1\otimes\cE_2$. Our findings are
graphically summarized in Fig.~\ref{fig:unital}. Based on this
figure one could argue that the property of being
entanglement-annihilating is not something rare and the noises
should be relatively close to unitary ones (or, equivalently,
sufficiently far from complete depolarization) in order to
guarantee the conservation of some entanglement.

\begin{figure}
\includegraphics[width=8.50cm]{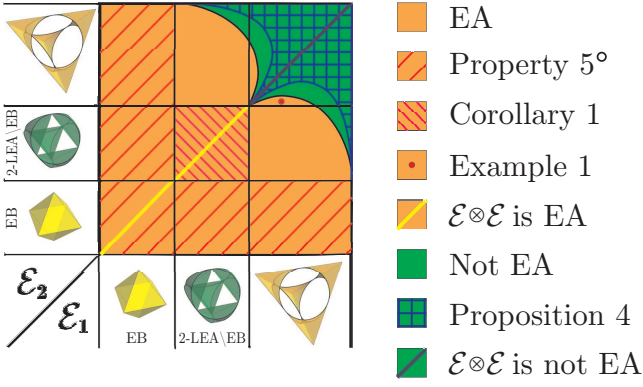}
\caption{\label{fig:unital} Schematic illustration of unital
two-qubit channels of the factorized form $\cE_1\otimes\cE_2$ with
respect to the entanglement-annihilating behavior. Each axis is
divided into three classes of unital single-qubit channels: EB,
2-LEA, and others.}
\end{figure}

\section{\label{sec:non-unital}Non-unital channels}

In this section we present a case study of a specific class
of non-unital qubit channels. In particular, we are interested which
of the statements valid in the case of unital channels can
be generalized also to non-unital case.

\subsection{\label{subsec:fact-extr-chann}Factorized extremal channels}

A channel is extremal if it cannot be expressed as a convex
combination of some other channels. A typical example is unitary
channels, which are the only extremal unital qubit channels. All
others are non-unital and include, for example, the families of
amplitude-damping channels $\cA_{p,0}$ and $\cA_{p,1}$ considered
in the next subsection.

Suppose constituents of the channel $\cE_1\otimes\cE_2$ can be
written as nontrivial convex sums of other channels, i.e. $\cE_j =
\mu_j \cF_{j} + (1-\mu_j) \cG_{j}$, $j=1,2$. Apparently,
$\cE_1\otimes\cE_2$ is then also a convex combination of
factorized channels $\cF_{j}\otimes\cG_{k}$ and
$\cG_{j}\otimes\cF_{k}$, $j,k=1,2$. If such factorized channels
are EA, then $\cE_1\otimes\cE_2$ is EA by property~$1^{\circ}$.
This analysis stimulates us to consider the extreme case of
channels $\cE_1\otimes\cE_2$, where neither $\cE_1$ nor $\cE_2$
cannot be resolved into nontrivial convex combinations.

As it was shown in the seminal paper~\cite{ruskai} by Ruskai et al., any
extremal qubit channel can be expressed in an apropriate basis as
a matrix
\begin{equation}
\label{extremal-parametrization}
\left(%
\begin{array}{cccc}
  1 & 0 & 0 & 0 \\
  0 & \cos u & 0 & 0 \\
  0 & 0 & \cos v & 0 \\
  \sin u \sin v & 0 & 0 & \cos u \cos v \\
\end{array}%
\right),
\end{equation}

\noindent where $u\in [0,2\pi)$ and $v\in [0,\pi)$.

At first we remind when an extremal qubit channel $\cE$ of the
form (\ref{extremal-parametrization}) is entanglement-breaking.
Let $\ket{\psi}$ be a maximally entangled two-qubit state.
Applying PPT criterion to the Choi-Jamio{\l}kowski state
$(\cE\otimes\cI)[\ket{\psi_{+}}\bra{\psi_{+}}]$, we find that
$\cE$ is EB if and only if $\cos u = 0$ or $\cos v = 0$.

In what follows we will use the fact that in formula
(\ref{extremal-parametrization}) the entries $\cos v$ and $\sin
u\sin v$ can simultaneously be made non-negative by an appropriate
choice of basis operators $\{\sigma_i\}_{i=0}^3$. The singular
values $\cos u$ and $\cos u \cos v$ have the same sign (either
positive or negative).

\begin{proposition}\label{prop:extremal2}
Suppose $\cE_1,\cE_2$ are extremal qubit channels. Then
$\cE_1\otimes\cE_2$ is EA if and only if either $\cE_1$ or $\cE_2$
is EB.
\end{proposition}
\begin{proof}
If either $\cE_1$ or $\cE_2$ is EB, then $\cE_1\otimes\cE_2$ is EA
in view of property~$5^{\circ}$. Let us now prove that the channel
$\cE_1\otimes\cE_2$ is not entanglement-annihilating if neither
$\cE_1$ nor $\cE_2$ is EB, i.e. $\cos u_1 \cos v_1 \cos u_2 \cos
v_2 \ne 0$, where parameters $(u_j,v_j)$ define the channel
$\cE_j$ according to formula (\ref{extremal-parametrization}).

Without loss of generality it can be assumed that $|\cos u_j| \ge
|\cos v_j|$, $j=1,2$. Consider the input state $\ket{\psi_{\rm
in}} = \frac{1}{\sqrt{2}} \big( \ket{\varphi\otimes\chi}+
\ket{\varphi_{\perp}\otimes\chi_{\perp}} \big)$, where the
orthogonal qubit states $\{ \ket{\varphi}, \ket{\varphi_{\perp}}
\}$ and $\{ \ket{\chi}, \ket{\chi_{\perp}} \}$ are parameterized
by angles $(\theta_1,\phi_1)$ and $(\theta_1,\phi_2)$,
respectively, as in Proposition~\ref{proposition1} (see formulas
(\ref{basis-vector})--(\ref{basis-vector-perp})). We put
$\phi_j=0$ and choose the angles $\theta_j$ such that
\begin{equation}
\cos\theta_j = \frac{\sin u_j \cos v_j}{\cos u_j \sin v_j} \,,
\quad \sin\theta_j = \frac{(\cos^2\! u_j - \cos^2\!
v_j)^{1/2}}{\cos u_j \sin v_j}, \nonumber
\end{equation}

\noindent then $\cE_1[\ket{\varphi}\bra{\varphi}]$ and
$\cE_2[\ket{\chi}\bra{\chi}]$ are known to be pure
states~\cite{ruskai}.

The reduction criterion (\ref{reduction-criterion}) guarantees
that the channel $\cE_1\otimes\cE_2$ is not EA if the operator
\begin{equation}
\label{M-operator} M=\Tr_2 \big[
(\cE_1\otimes\cE_2)[\ket{\psi_{\rm in}}\bra{\psi_{\rm in}}] \big]
\otimes I - (\cE_1\otimes\cE_2)[\ket{\psi_{\rm in}}\bra{\psi_{\rm
in}}]
\end{equation}

\noindent is not positive semi-definite. This takes place if there
exists a two-qubit state $\ket{\psi_{\rm test}}$ such that
$\bra{\psi_{\rm test}} M \ket{\psi_{\rm test}}<0$.

Let us construct a one-parametric family of candidates
$\ket{\psi_{p}}$, $p\in[0,1]$ for the state $\ket{\psi_{\rm
test}}$. To do that we use the pure states $\ket{\xi}\bra{\xi} =
\cE_1[\ket{\varphi}\bra{\varphi}]$ and $\ket{\zeta}\bra{\zeta} =
\cE_2[\ket{\chi}\bra{\chi}]$ and write $\ket{\psi_{p}} = \sqrt{p}
\ket{\xi\otimes\zeta} + \sqrt{1-p}
\ket{\xi_{\perp}\otimes\zeta_{\perp}}$. Direct calculation yields
\begin{eqnarray}
&& \!\!\!\!\!\!\!\!\!\!\!\! 4\bra{\psi_p} M \ket{\psi_p} = 1 -
\sin^2\! u_1 \sin^2\! u_2 - \cos^2\!
u_1 \cos^2\! u_2 \nonumber\\
&& \!\!\!\!\!\!\!\!\!\!\!\! - (1-2p)(\sin^2\! u_1 - \sin^2\! u_2)
- 4 \sqrt{p(1-p)} \Big( \cos
v_1 \cos v_2 \nonumber\\
&& \label{M-mean-value}\!\!\!\!\!\!\!\!\!\!\!\! + \sin u_1 \! \sin
u_2 \! \left[ (\cos^2\! u_1 \!-\! \cos^2\! v_1)(\cos^2\! u_2 \!-\!
\cos^2\! v_2) \right]^{\!1/2} \!\Big).
\end{eqnarray}

\noindent The state $\ket{\psi_{\rm test}}$ is then equal to such
$\ket{\psi_p}$ that minimizes the expression (\ref{M-mean-value}).
By a remark before the Proposition involved, $\sin u_j \ge 0$ and
$\cos v_j \ge 0$, $j=1,2$. Since $(1-2p)^2+(2\sqrt{p(1-p)})^2
\equiv 1$, the minimum of (\ref{M-mean-value}) is achievable and
reads
\begin{eqnarray}
&& \!\!\!\!\! 4\bra{\psi_{\rm test}} M \ket{\psi_{\rm test}} = 1 -
\sin^2\! u_1 \sin^2\! u_2 - \cos^2\!
u_1 \cos^2\! u_2 \nonumber\\
&& \!\!\!\!\! - \Big\{ (\sin^2\! u_1 - \sin^2\! u_2)^2 + 4 \Big(
\cos v_1 \cos v_2 \nonumber\\
&& \!\!\!\!\! + \sin u_1 \sin u_2 \! \left[ (\cos^2\! u_1 \!-\!
\cos^2\! v_1)(\cos^2\! u_2 \!-\! \cos^2\! v_2) \right]^{1/2}
\Big)^{\! 2} \Big\}^{\! 1/2}. \nonumber
\end{eqnarray}

It is not hard to see that the inequality $\bra{\psi_{\rm test}} M
\ket{\psi_{\rm test}} < 0$ is equivalent to
\begin{eqnarray}
&& \sin u_1 \sin u_2 \sqrt{(\cos^2\! u_1 \!-\! \cos^2\!
v_1)(\cos^2\! u_2 \!-\! \cos^2\! v_2)} \nonumber\\
&& + \cos v_1 \cos v_2 > |\sin u_1 \sin u_2 \cos u_1 \cos u_2|,
\nonumber
\end{eqnarray}

\noindent which is fulfilled for non-negative $\sin u_j$ and $\cos
v_j$ whenever $\cos u_1 \cos u_2 \cos v_1 \cos v_2 \ne 0$ because
$\sqrt{1-t^2}>1-t$ for all $t\in(0,1)$. Thus, the operator
(\ref{M-operator}) is not positive semi-definite and the channel
$\cE_1\otimes\cE_2$ is not EA.
\end{proof}

\begin{corollary}\label{prop:extremal1}
Suppose $\cE$ is an extremal qubit channel. Then $\cE\otimes\cE$
is EA, hence $\cE$ is 2-LEA, if and only if $\cE$ is EB.
\end{corollary}

Using results of Proposition~\ref{prop:extremal2}, we can present
a complete picture (Fig.~\ref{fig:picture-extremal}) of factorized
extremal channels analogues to that of unital channels
(Fig.~\ref{fig:unital}). Comparison of two figures gives a clear
insight that factorized extremal channels exhibit the best
possible entanglement preserving properties.

\begin{figure}
\includegraphics[width=8.50cm]{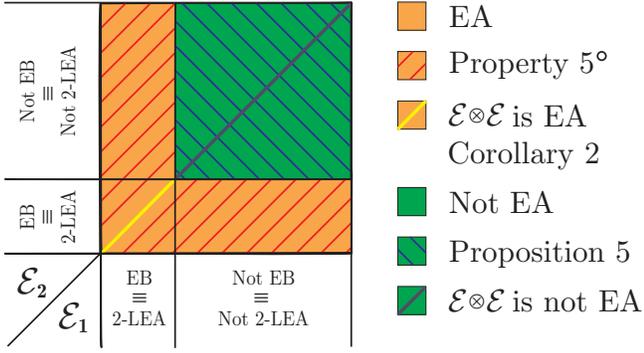}
\caption{\label{fig:picture-extremal} Schematic illustration of
factorized extremal two-qubit channels $\cE_1\otimes\cE_2$ with
respect to the entanglement-annihilating behavior. Each axis is
divided into two classes of extremal single-qubit channels: EB
that is equivalent to 2-LEA, and others.}
\end{figure}

\subsection{\label{subsec:amp-damping-channels}Generalized
amplitude-damping channels} Amplitude-damping channels
describe how a two-level system approaches the equilibrium due to
coupling with its environment, e.g., thanks to a spontaneous
emission process at zero temperature. If the environment has a
finite temperature, then such a dissipation process is described
by the action of a generalized amplitude-damping channel (see,
e.g.,~\cite{breuer}).

Kraus operators for a generalized amplitude-damping channel
$\cA_{p,\gamma}$ have the form
\begin{eqnarray}
&& E_0 = \sqrt{\gamma} \left(%
\begin{array}{cc}
  1 & 0 \\
  0 & \sqrt{1-p} \\
\end{array}%
\right), \qquad E_1 = \sqrt{\gamma} \left(%
\begin{array}{cc}
  0 & \sqrt{p} \\
  0 & 0 \\
\end{array}%
\right),\nonumber\\
&& E_2 = \sqrt{1-\gamma} \left(%
\begin{array}{cc}
  \sqrt{1-p} & 0 \\
  0 & 1 \\
\end{array}%
\right), E_3 = \sqrt{1-\gamma} \left(%
\begin{array}{cc}
  0 & 0 \\
  \sqrt{p} & 0 \\
\end{array}%
\right),\nonumber
\end{eqnarray}

\noindent where $p\in[0,1]$ determines the amplitude-damping rate
and $\gamma\in[0,1]$ is a parameter that depends on the
temperature and defines a fixed (equilibrium) state of
$\cA_{p,\gamma}$
$$\rho_{\infty} = \left(%
\begin{array}{cc}
  \gamma & 0 \\
  0 & 1-\gamma \\
\end{array}%
\right)\,.$$ If $\gamma=0$ or $\gamma=1$, then $\cA_{p,\gamma}$ is
simply an amplitude-damping channel. As it is mentioned above in
Sec.~\ref{subsec:fact-extr-chann}, amplitude-damping channels
$\cA_{p,0}$ and $\cA_{p,1}$ are extremal, thus obeying
Proposition~\ref{prop:extremal2}. A channel $\cA_{p_1,0} \otimes
\cA_{p_2,0}$ is EA if and only if $p_1=1$ or $p_2=1$, i.e. when at
least one of the constituent channels contracts the whole Bloch
sphere into a single pure state.

Let us now move on to generalized amplitude-damping channels
$\cA_{p,\gamma}$. We note that $\cA_{p,\gamma} = \gamma
\cA_{p,1}+(1-\gamma) \cA_{p,0}$. Since $\cA_{p,\gamma}$ is a
convex combinations of two amplitude-damping channels (extremal
qubit channels), we expect $\cA_{p_1,\gamma_1} \otimes
\cA_{p_2,\gamma_2}$ to exhibit worse entanglement preserving
properties than factorized extremal channels. We can expect that
the closer weights of two channels (the closer $\gamma$ to
$\frac{1}{2}$) the stronger the entanglement annihilation is.

\begin{figure}
\includegraphics[width=8.5cm]{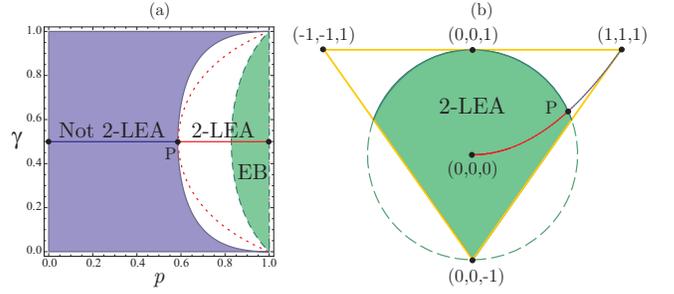}
\caption{\label{fig:amp-damping} (a) Regions of parameters $p$ and
$\gamma$, where the generalized amplitude-damping channel $\cE$ is
definitely not 2-LEA (violet) or $\cE$ is 2-LEA for sure (EB, red
solid line). Dotted curve presents a symbolic boundary between two
regions and seems to coincide with the solid curve. (b) 2-LEA
characterization of the generalized amplitude-damping channels
$\cE$ with $\gamma=1/2$. Such channels are unital and correspond
to infinite temperature. Family of these channels ($p\in[0,1]$) is
depicted among other unital channels in the cut
$\lambda_1=\lambda_2$ of the conventional reference frame
$(\lambda_1,\lambda_2,\lambda_3)$. Transition from 2-LEA to not
2-LEA behavior happens at point ${\rm
P}=(\sqrt{\sqrt{2}-1},\sqrt{\sqrt{2}-1},\sqrt{2}-1).$}
\end{figure}

Firstly, we note that the channel $\cA_{p,\gamma}$ is
entanglement-breaking if $p \ge (\sqrt{1 + 4\gamma (1-\gamma)} -
1) / 2 \gamma (1-\gamma)$. Let us remind that a channel is EB if
and only if its Choi-Jamio{\l}kowski state is separable and for
two qubits one can employ the PPT criterion to verify the
separability. In particular, for $\gamma=0$ (zero bath
temperature) the channel $\cA_p$ is EB only if $p=1$, hence, it
contracts the whole Bloch sphere into the equilibrium state.

Secondly, if we focus on 2-LEA channels, then it turns out that
$\cE\otimes\cE$ preserves entanglement of the state
$\ket{\psi_{+}}$ if $p \le
(1-\sqrt{2\gamma(1-\gamma)})/(1-2\gamma(1-\gamma))$. These results
are shown in Fig.~\ref{fig:amp-damping}a, where the dotted line
depicts a symbolic boundary between EA and not EA channels. The
numerical analysis encourages to assume that this boundary
coincides with the solid line.

Thirdly, if $\gamma = 1/2$, then $\cE$ becomes unital. In this
case to conclude EA of the channel $\cE\otimes\cE$ it suffices to
consider its action of on the maximally entangled state, e.g.,
$\ket{\psi_{+}}$. Illustration of such unital channels with
respect to other unital channels is presented in
Fig.~\ref{fig:amp-damping}b.

Finally, we can draw a conclusion that a picture of the
entanglement-annihilation behavior of generalized
amplitude-damping channels takes an intermediate form between the
unitary two-qubit channels (Fig.~\ref{fig:unital}) and the
factorized extremal channel (Fig.~\ref{fig:picture-extremal}).

\section{EA-extremal local channels}\label{sec:extremal}

Convexity is an important property of many channel sets. Thus, the
set of all qubit channels ${\sf T}_2$ and the set of all EB qubit
channels ${\sf T}_{\rm EB}$ are convex. Extreme points of these
sets are studied in Refs.~\cite{ruskai} and
\cite{horodecki-shor-ruskai}, respectively. Extreme points of the
set of EB channels are referred to as EB-extremal. It is shown in
the paper~\cite{horodecki-shor-ruskai}, that in qubit case all
EB-extremal single-qubit channels can be represented in the form
\begin{equation}
\label{EB-qubit-channel} \cF [\rho] = \bra{\psi} \rho \ket{\psi}
\ket{\varphi_1}\bra{\varphi_1} + \bra{\psi_{\perp}} \rho
\ket{\psi_{\perp}} \ket{\varphi_2}\bra{\varphi_2},
\end{equation}

\noindent where
$\ket{\psi},\ket{\psi_{\perp}},\ket{\varphi_1},\ket{\varphi_2}$
are pure qubit states and $\ip{\psi}{\psi_{\perp}} = 0$.

If we consider general two-qubit channels (not necessarily local),
then along with the set of all channels ${\sf T}_{\rm chan}$ and
the set of EB two-qubit channels there is a set of all EA channels
${\sf T}_{\rm EA}$ which is also convex (property~$1^{\circ}$).
The question arises itself to find extreme points of this set,
i.e. channels that are extreme for two-qubit EA channels. We will
refer to such channels as EA-extremal. The following proposition
provides EA-extremal channels of factorized form.

\begin{proposition} \label{prop:EA-extremal}
Let $\cF$ be EB-extremal one-qubit channel, then the channels
$\cF\otimes \cI$, $\cI\otimes\cF$ are EA-extremal.
\end{proposition}

\begin{proof}
Let us now prove the statement of Proposition by reductio ad
absurdum. Suppose $\cF\otimes \cI$ is not EA-extremal, i.e. there
exist two non-coincident EA channels $\cG_1$ and $\cG_2$ such that
$\cF\otimes \cI = \mu \cG_1 + (1-\mu) \cG_2$, where $\mu\in[0,1]$.
Using the Choi-Jamio{\l}kowski
isomorphism~\cite{choi,jamiolkowski}, the latter equation can be
rewritten as $\rho_{\cF}\otimes P_+ = \mu \omega_1 + (1-\mu)
\omega_2$, where $P_+$ is a 1-rank projector onto maximally
entangled state $\ket{\psi_+}$. Taking partial trace over the
first subsystem $A$, we get $P_+ = \mu \Tr_{A}[\omega_1] + (1-\mu)
\Tr_{A}[\omega_2]$, from which it follows that $\Tr_{A}[\omega_1]
= \Tr_{A}[\omega_2] = P_+$. Consequently, $\omega_j = \rho_j
\otimes P_+$, $j=1,2$, because if a subsystem is in a pure state,
then it is necessarily factorized from any other system. As a
result we may conclude that $\cF\otimes \cI = (\mu \cE_1 + (1-\mu)
\cE_2)\otimes \cI$, where $\cE_j\otimes\cI$ is EA, i.e. $\cE_j$ is
EB, $j=1,2$ (property $6^{\circ}$). Since $\cF$ is EB-extremal,
then the relation $\cF = \mu \cE_1 + (1-\mu) \cE_2$ can be only
fulfilled if $\cE_1 = \cE_2 = \cF$, i.e. $\cG_1=\cG_2$. This
contradiction concludes the proof.
\end{proof}

Let us stress two points. Firstly, the identity channel in the
proposition can be replaced by any unitary channel. Secondly, this
proposition is not a special case of
Proposition~\ref{prop:extremal2}, because there we assume channels
$\cE_1$ and $\cE_2$ to be extremal in the set of all qubit
channels ${\sf T}_2$, whereas in
Proposition~\ref{prop:EA-extremal} the local channel $\cF$ is
assumed to be EB-extremal. In particular, if both $\cE_1,\cE_2$
are extremal and EB, then $\cE_1\otimes\cE_2$ is simultaneuously
extremal, EB-extremal and EA-extremal. The following example
demonstrates the existence of EA-extremal channels that are not
extreme points of the set of all two-qubit channels ${\sf T}_{\rm
chan}$.

\begin{example} \label{example:EA-extremal}
Consider a phase damping channel
$\cF[\varrho]=\tfrac{1}{2}(\varrho+\sigma_z\varrho\sigma_z)$. Clearly, this
channel is unital and not extremal in the set of all channels. Nevertheless,
its action can be expressed also as
$$\cF[\varrho]=
\bra{0} \rho \ket{0}
\ket{0}\bra{0} + \bra{1} \rho
\ket{1} \ket{1}\bra{1}\,.
$$
This means the phase-damping channel is an extreme point of the
set of entanglement-breaking channels, hence, it is
EB-extremal~\cite{horodecki-shor-ruskai}. Then by
Proposition~\ref{prop:EA-extremal} the two-qubit channels
$\cF\otimes\cI$ and $\cI \otimes \cF$ are EA-extremal.
\end{example}

\section{\label{sec:duality}EA-duality}

In this section we will employ the idea of
entanglement-annihilation to introduce the concept of EA-duality
for local channels. Let us denote by ${\sf T}_d$ the set of
channels on a $d$-dimensional (non-composite) quantum system
(qudit). Let ${\sf Q}\subset {\sf T}_d$ be an arbitrary subset of
qudit channels. A subset $\widetilde{{\sf Q}}_n\subset{\sf T}_n$
is called $n$-EA-dual to ${\sf Q}$ if $\cE_1\otimes\cE_2$ is EA
for all $\cE_1\in{\sf Q}$ and $\cE_2\in\widetilde{{\sf Q}}_n$. To
be more precise,
\begin{equation}
\widetilde{{\sf Q}}_n = \{\cE_2\in{\sf T}_n ~ | ~
\cE_1\otimes\cE_2 ~{\rm is~EA~for~all}~ \cE_1\in{\sf Q}\}.
\end{equation}
When it is clear from the context we will omit the explicit mentioning
of the dimensions and  say $\widetilde{{\sf Q}}$
is EA-dual to ${\sf Q}$.

Using this concept, we may rephrase the goal of this paper as an
identification of the set of channels ${\sf T}_{\rm
ent}\subset{\sf T}_2$ such that EA-duals of all its elements
coincide with the set of entanglement-breaking channels, i.e.
$\widetilde{\{\cE\}}={\sf T}_{\rm EB}$ for all $\cE\in{\sf T}_{\rm
ent}$. Let us note that for any ${\sf Q}$ its EA-dual
$\widetilde{\sf Q}$ contains the entanglement-breaking channels,
i.e. ${\sf T}_{\rm EB}\subset\widetilde{{\sf Q}}$. It is also
clear that $\widetilde{\sf Q}\subset\widetilde{\{\cE\}}$ providing
that $\cE\in{\sf Q}$.

\begin{example}
Let us clarify the introduced concept on some examples:
\begin{itemize}
\item{} For unitary channels ${\cal U}[\varrho]=U\varrho U^*$ the EA-dual
coincides with EB qubit channels, i.e. $\widetilde{\{{\cal U}\}}
={\sf T}_{\rm EB}\subset{\sf T}_n$,
which means that the system under consideration is maximally robust
in sharing the entanglement unless an entanglement-breaking channel
is applied on the second system.

\item For the set ${\sf T}_{\rm EB}$ of entanglement-breaking
channels the EA-dual set equals to the set of all channels, i.e.
$\widetilde{{\sf T}}_{\rm EB}={\sf T}_{n}$. This only illustrates
the fact that entanglement-breaking channels destroy any
entanglement between the system under consideration and whatever
second system. For each entanglement-breaking channel $\cE$ its
EA-dual contains all channels, i.e. $\widetilde{\{\cE\}}={\sf
T}_{n}$.

\item The EA-dual of all channels is the set of entanglement-breaking
channels, i.e. $\widetilde{{\sf T}}_{d}={\sf T}_{\rm EB}\subset{\sf T}_n$.
Providing that $d=n$ we can write $\widetilde{\widetilde{{\sf T}}}_{d}
= {\sf T}_{d}$.

\end{itemize}
\end{example}

Fixing $d$ and $n$ the dual sets obviously satisfy the
following properties:
\begin{eqnarray}
&& {\sf Q}_1 \subset {\sf Q}_2 \Rightarrow \widetilde{{\sf Q}}_1
\supset
\widetilde{{\sf Q}}_2,\\
&& \label{duality-property-equality} \widetilde{({\sf Q}_1 \cup
{\sf Q}_2)} = \widetilde{{\sf Q}}_1 \cap
\widetilde{{\sf Q}}_2,\\
&& \widetilde{({\sf Q}_1 \cap {\sf Q}_2)} \supset \widetilde{{\sf
Q}}_1 \cup \widetilde{{\sf Q}}_2.
\end{eqnarray}

The results presented in the previous sections contain partial
answers to EA-duality sets for 2-LEA qubit channels, or
depolarizing qubit channels, or amplitude-damping qubit channels,
etc. However, further analysis is beyond the scope of this paper
and the full characterization of EA-duals of these and other
interesting subsets of qubit channels remains an open problem.

\section{\label{sec:conclusions}Summary}

In this paper we have investigated the robustness of entanglement
in two-qubit (spatially separated) systems under the influence of
independent reservoirs. In particular, we paid attention to
characterization of the so-called entanglement-annihlating (EA)
channels. These are the channels that completely annihilate any
entanglement initially present between the subsystems. In
contrast, the so-called entanglement-breaking channels (EB)
destroy any entanglement between the system they act on (two-qubit
system in our case) and any other system. The dramatic
differences, but also subtle relations, between these two concepts
were discovered for particular classes of channels.

We succeeded to characterize all unital two-qubits
entanglement-annihilating channels of the
factorized form $\cE_1\otimes\cE_2$ and
the results are nicely illustrated in Fig.~\ref{fig:unital}. We
derived a sufficient condition (Proposition~\ref{proposition4})
for a local unital channel not to be EA, which guaranties the
entanglement retention and enables to
perform entanglement-enabled experiments in the presence of such
noise. We gave (Proposition~\ref{proposition2}) the complete
characterization of unital
entanglement-annihilating channels $\cE\otimes\cE$ and showed
(Proposition~\ref{proposition3})
that such (2-LEA) channels $\cE$ form a convex subset of the set of all
unital single-qubit channels.

For example, in case of depolarizing channels $\cE_1$ and $\cE_2$
with rates $q_1$ and $q_2$, their product should be kept above
$\frac{1}{3}$. Above this critical value, there still exist
initial states of two-qubit system for which the entanglement
survives the effects of depolarizing noise.

Particular results have been also obtained for the case of extremal
non-unital channels, for which the entanglement turns out to be
more robust (in comparison with unital channels). In particular,
they are entanglement-annihilating only if one of the constituent
channels is entanglement-breaking meaning that (in this case)
the set of 2-LEA channels coincides with EB channels.

Much attention has been also focused on such a fundamental entity
as convexity of the set of all two-qubit channels ${\sf T}_{\rm
chan}$ and the set of EA channels ${\sf T}_{\rm EA}$. We have
revealed that the sets ${\sf T}_{\rm chan}$ and ${\sf T}_{\rm EA}$
have common extremal points corresponding to local channels. We
have constructed a class of purely EA-extremal channels
(phase-damping channels), i.e. channels that are extremal for
${\sf T}_{\rm EA}$ but are internal for ${\sf T}_{\rm chan}$.

Finally, we have introduced an important concept of EA-duality
between sets of channels defined on individual subsystems. This
concept gives another perspective on classification of channels
with respect to their entanglement annihilation `potential'.

To conclude, the presented analysis contains partial
characterization of local entanglement-annihilation two-qubit
channels. This class of channels is of particular importance and
interest in the domain of quantum information processing. Although
the complete understanding of the phenomena of
entanglement-annihilation is still missing, the presented results
represent important and practical steps towards this direction.
Our analysis shows that the phenomenon of being
entanglement-annihilating is not a rare one and, hence, deserves
further attention. From the practical point of view, we have
analyzed in detail classes of physically relevant qubit channels
(depolarizing, phase-damping, amplitude-damping) and among them
identified the `good' and `bad' ones.

\begin{acknowledgments}
This research was carried out while S.N.F. was visiting at the
Research Center for Quantum Information, Institute of Physics,
Slovak Academy of Sciences. S.N.F. is grateful for very kind
hospitality. This work was supported by EU integrated project
2010-248095 (Q-ESSENCE), APVV DO7RP-0002-10 and VEGA 2/0092/09 (QWAEN).
S.N.F. thanks the Russian Foundation for Basic Research (projects
09-02-00142, 10-02-00312, and 11-02-00456), the Russian Science
Support Foundation, the Dynasty Foundation, and the Ministry of
Education and Science of the Russian Federation (projects
2.1.1/5909, $\Pi$558, 14.740.11.0497, and 14.740.11.1257).
T.R. thanks to APVV LPP-0264-07 (QWOSSI) and
M.Z. acknowledges the support of SCIEX Fellowship 10.271.
\end{acknowledgments}

\bibliography{EA}

\begin{thebibliography}{32}
\expandafter\ifx\csname natexlab\endcsname\relax\def\natexlab#1{#1}\fi
\expandafter\ifx\csname bibnamefont\endcsname\relax
  \def\bibnamefont#1{#1}\fi
\expandafter\ifx\csname bibfnamefont\endcsname\relax
  \def\bibfnamefont#1{#1}\fi
\expandafter\ifx\csname citenamefont\endcsname\relax
  \def\citenamefont#1{#1}\fi
\expandafter\ifx\csname url\endcsname\relax
  \def\url#1{\texttt{#1}}\fi
\expandafter\ifx\csname urlprefix\endcsname\relax\def\urlprefix{URL }\fi
\providecommand{\bibinfo}[2]{#2}
\providecommand{\eprint}[2][]{\url{#2}}

\bibitem[{\citenamefont{Schr\"{o}dinger}(1935)}]{schrodinger}
\bibinfo{author}{\bibfnamefont{E.}~\bibnamefont{Schr\"{o}dinger}},
  \bibinfo{journal}{Mathematical Proceedings of the Cambridge Philosophical
  Society} \textbf{\bibinfo{volume}{31}}, \bibinfo{pages}{555}
  (\bibinfo{year}{1935}).

\bibitem[{\citenamefont{Horodecki et~al.}(2009)\citenamefont{Horodecki,
  Horodecki, Horodecki, and Horodecki}}]{horodecki-review}
\bibinfo{author}{\bibfnamefont{R.}~\bibnamefont{Horodecki}},
  \bibinfo{author}{\bibfnamefont{P.}~\bibnamefont{Horodecki}},
  \bibinfo{author}{\bibfnamefont{M.}~\bibnamefont{Horodecki}},
  \bibnamefont{and}
  \bibinfo{author}{\bibfnamefont{K.}~\bibnamefont{Horodecki}},
  \bibinfo{journal}{Rev. Mod. Phys.} \textbf{\bibinfo{volume}{81}},
  \bibinfo{pages}{865} (\bibinfo{year}{2009}).

\bibitem[{\citenamefont{\ifmmode~\dot{Z}\else \.{Z}\fi{}yczkowski
  et~al.}(2001)\citenamefont{\ifmmode~\dot{Z}\else \.{Z}\fi{}yczkowski,
  Horodecki, Horodecki, and Horodecki}}]{zyczkowski}
\bibinfo{author}{\bibfnamefont{K.}~\bibnamefont{\ifmmode~\dot{Z}\else
  \.{Z}\fi{}yczkowski}},
  \bibinfo{author}{\bibfnamefont{P.}~\bibnamefont{Horodecki}},
  \bibinfo{author}{\bibfnamefont{M.}~\bibnamefont{Horodecki}},
  \bibnamefont{and}
  \bibinfo{author}{\bibfnamefont{R.}~\bibnamefont{Horodecki}},
  \bibinfo{journal}{Phys. Rev. A} \textbf{\bibinfo{volume}{65}},
  \bibinfo{pages}{012101} (\bibinfo{year}{2001}).

\bibitem[{\citenamefont{Sinayskiy et~al.}(2009)\citenamefont{Sinayskiy,
  Ferraro, Napoli, Messina, and Petruccione}}]{sinayskiy}
\bibinfo{author}{\bibfnamefont{I.}~\bibnamefont{Sinayskiy}},
  \bibinfo{author}{\bibfnamefont{E.}~\bibnamefont{Ferraro}},
  \bibinfo{author}{\bibfnamefont{A.}~\bibnamefont{Napoli}},
  \bibinfo{author}{\bibfnamefont{A.}~\bibnamefont{Messina}}, \bibnamefont{and}
  \bibinfo{author}{\bibfnamefont{F.}~\bibnamefont{Petruccione}},
  \bibinfo{journal}{Journal of Physics A: Mathematical and Theoretical}
  \textbf{\bibinfo{volume}{42}}, \bibinfo{pages}{485301}
  (\bibinfo{year}{2009}).

\bibitem[{\citenamefont{Shan et~al.}(2009)\citenamefont{Shan, Liu, Cheng, Liu,
  Huang, and Li}}]{shan}
\bibinfo{author}{\bibfnamefont{C.-J.} \bibnamefont{Shan}},
  \bibinfo{author}{\bibfnamefont{J.-B.} \bibnamefont{Liu}},
  \bibinfo{author}{\bibfnamefont{W.-W.} \bibnamefont{Cheng}},
  \bibinfo{author}{\bibfnamefont{T.-K.} \bibnamefont{Liu}},
  \bibinfo{author}{\bibfnamefont{Y.-X.} \bibnamefont{Huang}}, \bibnamefont{and}
  \bibinfo{author}{\bibfnamefont{H.}~\bibnamefont{Li}},
  \bibinfo{journal}{Communications in Theoretical Physics}
  \textbf{\bibinfo{volume}{51}}, \bibinfo{pages}{1013} (\bibinfo{year}{2009}).

\bibitem[{\citenamefont{Scala et~al.}(2011)\citenamefont{Scala, Migliore,
  Messina, and S\'{a}nchez-Soto}}]{scala}
\bibinfo{author}{\bibfnamefont{M.}~\bibnamefont{Scala}},
  \bibinfo{author}{\bibfnamefont{R.}~\bibnamefont{Migliore}},
  \bibinfo{author}{\bibfnamefont{A.}~\bibnamefont{Messina}}, \bibnamefont{and}
  \bibinfo{author}{\bibfnamefont{L.~L.} \bibnamefont{S\'{a}nchez-Soto}},
  \bibinfo{journal}{The European Physical Journal D - Atomic, Molecular,
  Optical and Plasma Physics} \textbf{\bibinfo{volume}{61}},
  \bibinfo{pages}{199} (\bibinfo{year}{2011}).

\bibitem[{\citenamefont{Cui et~al.}(2009)\citenamefont{Cui, Xi, and Pan}}]{cui}
\bibinfo{author}{\bibfnamefont{W.}~\bibnamefont{Cui}},
  \bibinfo{author}{\bibfnamefont{Z.}~\bibnamefont{Xi}}, \bibnamefont{and}
  \bibinfo{author}{\bibfnamefont{Y.}~\bibnamefont{Pan}},
  \bibinfo{journal}{Journal of Physics A: Mathematical and Theoretical}
  \textbf{\bibinfo{volume}{42}}, \bibinfo{pages}{155303}
  (\bibinfo{year}{2009}).

\bibitem[{\citenamefont{Altintas and Eryigit}(2010)}]{altintas}
\bibinfo{author}{\bibfnamefont{F.}~\bibnamefont{Altintas}} \bibnamefont{and}
  \bibinfo{author}{\bibfnamefont{R.}~\bibnamefont{Eryigit}},
  \bibinfo{journal}{Journal of Physics A: Mathematical and Theoretical}
  \textbf{\bibinfo{volume}{43}}, \bibinfo{pages}{415306}
  (\bibinfo{year}{2010}).

\bibitem[{\citenamefont{Ferraro et~al.}(2010)\citenamefont{Ferraro, Scala,
  Migliore, and Napoli}}]{ferraro}
\bibinfo{author}{\bibfnamefont{E.}~\bibnamefont{Ferraro}},
  \bibinfo{author}{\bibfnamefont{M.}~\bibnamefont{Scala}},
  \bibinfo{author}{\bibfnamefont{R.}~\bibnamefont{Migliore}}, \bibnamefont{and}
  \bibinfo{author}{\bibfnamefont{A.}~\bibnamefont{Napoli}},
  \bibinfo{journal}{Physica Scripta} \textbf{\bibinfo{volume}{T140}},
  \bibinfo{pages}{014042} (\bibinfo{year}{2010}).

\bibitem[{\citenamefont{Li et~al.}(2010)\citenamefont{Li, Zou, and
  Shao}}]{li2010}
\bibinfo{author}{\bibfnamefont{J.-G.} \bibnamefont{Li}},
  \bibinfo{author}{\bibfnamefont{J.}~\bibnamefont{Zou}}, \bibnamefont{and}
  \bibinfo{author}{\bibfnamefont{B.}~\bibnamefont{Shao}},
  \bibinfo{journal}{Phys. Rev. A} \textbf{\bibinfo{volume}{82}},
  \bibinfo{pages}{042318} (\bibinfo{year}{2010}).

\bibitem[{\citenamefont{Li et~al.}(2011)\citenamefont{Li, Zou, and
  Shao}}]{li2011}
\bibinfo{author}{\bibfnamefont{J.-G.} \bibnamefont{Li}},
  \bibinfo{author}{\bibfnamefont{J.}~\bibnamefont{Zou}}, \bibnamefont{and}
  \bibinfo{author}{\bibfnamefont{B.}~\bibnamefont{Shao}},
  \bibinfo{journal}{Physics Letters A} \textbf{\bibinfo{volume}{375}},
  \bibinfo{pages}{2300 } (\bibinfo{year}{2011}).

\bibitem[{\citenamefont{Merkli et~al.}(2011)\citenamefont{Merkli, Berman,
  Borgonovi, and Gebresellasie}}]{merkli}
\bibinfo{author}{\bibfnamefont{M.}~\bibnamefont{Merkli}},
  \bibinfo{author}{\bibfnamefont{G.}~\bibnamefont{Berman}},
  \bibinfo{author}{\bibfnamefont{F.}~\bibnamefont{Borgonovi}},
  \bibnamefont{and}
  \bibinfo{author}{\bibfnamefont{K.}~\bibnamefont{Gebresellasie}},
  \emph{\bibinfo{title}{Evolution of entanglement of two qubits interacting
  through local and collective environments}} (\bibinfo{year}{2011}),
  \eprint{arXiv:1001.1144v1 [quant-ph]}.

\bibitem[{\citenamefont{Almeida et~al.}(2007)\citenamefont{Almeida, de~Melo,
  Hor-Meyll, Salles, Walborn, Ribeiro, and Davidovich}}]{almeida}
\bibinfo{author}{\bibfnamefont{M.~P.} \bibnamefont{Almeida}},
  \bibinfo{author}{\bibfnamefont{F.}~\bibnamefont{de~Melo}},
  \bibinfo{author}{\bibfnamefont{M.}~\bibnamefont{Hor-Meyll}},
  \bibinfo{author}{\bibfnamefont{A.}~\bibnamefont{Salles}},
  \bibinfo{author}{\bibfnamefont{S.~P.} \bibnamefont{Walborn}},
  \bibinfo{author}{\bibfnamefont{P.~H.~S.} \bibnamefont{Ribeiro}},
  \bibnamefont{and}
  \bibinfo{author}{\bibfnamefont{L.}~\bibnamefont{Davidovich}},
  \bibinfo{journal}{Science} \textbf{\bibinfo{volume}{316}},
  \bibinfo{pages}{579} (\bibinfo{year}{2007}).

\bibitem[{\citenamefont{Yu and Eberly}(2009)}]{yu}
\bibinfo{author}{\bibfnamefont{T.}~\bibnamefont{Yu}} \bibnamefont{and}
  \bibinfo{author}{\bibfnamefont{J.~H.} \bibnamefont{Eberly}},
  \bibinfo{journal}{Science} \textbf{\bibinfo{volume}{323}},
  \bibinfo{pages}{598} (\bibinfo{year}{2009}).

\bibitem[{\citenamefont{de~Brito and Ramos}(2006)}]{deBrito}
\bibinfo{author}{\bibfnamefont{W.~A.} \bibnamefont{de~Brito}} \bibnamefont{and}
  \bibinfo{author}{\bibfnamefont{R.~V.} \bibnamefont{Ramos}},
  \bibinfo{journal}{Physics Letters A} \textbf{\bibinfo{volume}{360}},
  \bibinfo{pages}{251 } (\bibinfo{year}{2006}).

\bibitem[{\citenamefont{Tiersch et~al.}(2008)\citenamefont{Tiersch, de~Melo,
  and Buchleitner}}]{tiersch}
\bibinfo{author}{\bibfnamefont{M.}~\bibnamefont{Tiersch}},
  \bibinfo{author}{\bibfnamefont{F.}~\bibnamefont{de~Melo}}, \bibnamefont{and}
  \bibinfo{author}{\bibfnamefont{A.}~\bibnamefont{Buchleitner}},
  \bibinfo{journal}{Phys. Rev. Lett.} \textbf{\bibinfo{volume}{101}},
  \bibinfo{pages}{170502} (\bibinfo{year}{2008}).

\bibitem[{\citenamefont{Konrad et~al.}(2008)\citenamefont{Konrad, de~Melo,
  Tiersch, Kasztelan, {a}o, and Buchleitner}}]{konrad}
\bibinfo{author}{\bibfnamefont{T.}~\bibnamefont{Konrad}},
  \bibinfo{author}{\bibfnamefont{F.}~\bibnamefont{de~Melo}},
  \bibinfo{author}{\bibfnamefont{M.}~\bibnamefont{Tiersch}},
  \bibinfo{author}{\bibfnamefont{C.}~\bibnamefont{Kasztelan}},
  \bibinfo{author}{\bibfnamefont{A.~A.} \bibnamefont{{a}o}}, \bibnamefont{and}
  \bibinfo{author}{\bibfnamefont{A.}~\bibnamefont{Buchleitner}},
  \bibinfo{journal}{Nature Physics} \textbf{\bibinfo{volume}{4}},
  \bibinfo{pages}{99} (\bibinfo{year}{2008}).

\bibitem[{\citenamefont{Zhang et~al.}(2010)\citenamefont{Zhang, Luo, Ren, and
  Sun}}]{zhang2010}
\bibinfo{author}{\bibfnamefont{H.}~\bibnamefont{Zhang}},
  \bibinfo{author}{\bibfnamefont{J.}~\bibnamefont{Luo}},
  \bibinfo{author}{\bibfnamefont{T.-T.} \bibnamefont{Ren}}, \bibnamefont{and}
  \bibinfo{author}{\bibfnamefont{X.-P.} \bibnamefont{Sun}},
  \bibinfo{journal}{Chin. Phys. Lett.} \textbf{\bibinfo{volume}{27}},
  \bibinfo{pages}{090303} (\bibinfo{year}{2010}).

\bibitem[{\citenamefont{Morav\v{c}\'{i}kov\'{a} and
  Ziman}(2010)}]{moravchikova}
\bibinfo{author}{\bibfnamefont{L.}~\bibnamefont{Morav\v{c}\'{i}kov\'{a}}}
  \bibnamefont{and} \bibinfo{author}{\bibfnamefont{M.}~\bibnamefont{Ziman}},
  \bibinfo{journal}{Journal of Physics A: Mathematical and Theoretical}
  \textbf{\bibinfo{volume}{43}}, \bibinfo{pages}{275306}
  (\bibinfo{year}{2010}).

\bibitem[{\citenamefont{Holevo}(1998)}]{holevo}
\bibinfo{author}{\bibfnamefont{A.~S.} \bibnamefont{Holevo}},
  \bibinfo{journal}{Russian Mathematical Surveys}
  \textbf{\bibinfo{volume}{53}}, \bibinfo{pages}{1295} (\bibinfo{year}{1998}).

\bibitem[{\citenamefont{King}(2002)}]{king}
\bibinfo{author}{\bibfnamefont{C.}~\bibnamefont{King}}, \bibinfo{journal}{J.
  Math. Phys.} \textbf{\bibinfo{volume}{43}}, \bibinfo{pages}{1247}
  (\bibinfo{year}{2002}).

\bibitem[{\citenamefont{Shor}(2002)}]{shor}
\bibinfo{author}{\bibfnamefont{P.~W.} \bibnamefont{Shor}}, \bibinfo{journal}{J.
  Math. Phys.} \textbf{\bibinfo{volume}{43}}, \bibinfo{pages}{4334}
  (\bibinfo{year}{2002}).

\bibitem[{\citenamefont{Ruskai}(2002)}]{ruskai-eb}
\bibinfo{author}{\bibfnamefont{M.~B.} \bibnamefont{Ruskai}},
  \emph{\bibinfo{title}{Entanglement breaking channels}}
  (\bibinfo{year}{2002}), \eprint{arXiv:quant-ph/0207100v2}.

\bibitem[{\citenamefont{Horodecki et~al.}(2003)\citenamefont{Horodecki, Shor,
  and Ruskai}}]{horodecki-shor-ruskai}
\bibinfo{author}{\bibfnamefont{M.}~\bibnamefont{Horodecki}},
  \bibinfo{author}{\bibfnamefont{P.~W.} \bibnamefont{Shor}}, \bibnamefont{and}
  \bibinfo{author}{\bibfnamefont{M.~B.} \bibnamefont{Ruskai}},
  \bibinfo{journal}{Reviews in Mathematical Physics}
  \textbf{\bibinfo{volume}{15}}, \bibinfo{pages}{629} (\bibinfo{year}{2003}).

\bibitem[{\citenamefont{Holevo}(2008)}]{holevo08}
\bibinfo{author}{\bibfnamefont{A.~S.} \bibnamefont{Holevo}},
  \bibinfo{journal}{Problems of Information Transmission}
  \textbf{\bibinfo{volume}{44}}, \bibinfo{pages}{3} (\bibinfo{year}{2008}).

\bibitem[{\citenamefont{Peres}(1996)}]{peres}
\bibinfo{author}{\bibfnamefont{A.}~\bibnamefont{Peres}},
  \bibinfo{journal}{Phys. Rev. Lett.} \textbf{\bibinfo{volume}{77}},
  \bibinfo{pages}{1413} (\bibinfo{year}{1996}).

\bibitem[{\citenamefont{Horodecki et~al.}(1996)\citenamefont{Horodecki,
  Horodecki, and Horodecki}}]{horodecki96}
\bibinfo{author}{\bibfnamefont{M.}~\bibnamefont{Horodecki}},
  \bibinfo{author}{\bibfnamefont{P.}~\bibnamefont{Horodecki}},
  \bibnamefont{and}
  \bibinfo{author}{\bibfnamefont{R.}~\bibnamefont{Horodecki}},
  \bibinfo{journal}{Physics Letters A} \textbf{\bibinfo{volume}{223}},
  \bibinfo{pages}{1 } (\bibinfo{year}{1996}).

\bibitem[{\citenamefont{Ruskai et~al.}(2002)\citenamefont{Ruskai, Szarek, and
  Werner}}]{ruskai}
\bibinfo{author}{\bibfnamefont{M.~B.} \bibnamefont{Ruskai}},
  \bibinfo{author}{\bibfnamefont{S.}~\bibnamefont{Szarek}}, \bibnamefont{and}
  \bibinfo{author}{\bibfnamefont{E.}~\bibnamefont{Werner}},
  \bibinfo{journal}{Linear Algebra and its Applications}
  \textbf{\bibinfo{volume}{347}}, \bibinfo{pages}{159 } (\bibinfo{year}{2002}).

\bibitem[{\citenamefont{Horodecki and Horodecki}(1999)}]{horodecki99}
\bibinfo{author}{\bibfnamefont{M.}~\bibnamefont{Horodecki}} \bibnamefont{and}
  \bibinfo{author}{\bibfnamefont{P.}~\bibnamefont{Horodecki}},
  \bibinfo{journal}{Phys. Rev. A} \textbf{\bibinfo{volume}{59}},
  \bibinfo{pages}{4206} (\bibinfo{year}{1999}).

\bibitem[{\citenamefont{Breuer and Petruccione}(2002)}]{breuer}
\bibinfo{author}{\bibfnamefont{H.-P.} \bibnamefont{Breuer}} \bibnamefont{and}
  \bibinfo{author}{\bibfnamefont{F.}~\bibnamefont{Petruccione}},
  \emph{\bibinfo{title}{The Theory of open quantum systems}}
  (\bibinfo{publisher}{Oxford University Press}, \bibinfo{year}{2002}).

\bibitem[{\citenamefont{Choi}(1975)}]{choi}
\bibinfo{author}{\bibfnamefont{M.-D.} \bibnamefont{Choi}},
  \bibinfo{journal}{Linear Algebra and its Applications}
  \textbf{\bibinfo{volume}{10}}, \bibinfo{pages}{285 } (\bibinfo{year}{1975}).

\bibitem[{\citenamefont{Jamio{\l}kowski}(1972)}]{jamiolkowski}
\bibinfo{author}{\bibfnamefont{A.}~\bibnamefont{Jamio{\l}kowski}},
  \bibinfo{journal}{Reports on Mathematical Physics}
  \textbf{\bibinfo{volume}{3}}, \bibinfo{pages}{275 } (\bibinfo{year}{1972}).

\end{thebibliography}

\end{document}